\documentclass{sig-alternate}

\usepackage[noend]{algorithmic,algorithm}
\usepackage{multicol,multirow}
\usepackage{amsmath,amsfonts,amssymb}
\usepackage{float,xfrac}
\usepackage{color,graphicx,graphics,epsfig,epstopdf}
\usepackage[dvipsnames]{xcolor}
\usepackage{cite,url}

\newtheorem{theorem}{Theorem}[section]
\newtheorem{corollary}{Corollary}
\newtheorem{lemma}{Lemma}
\newtheorem{claim}{Claim}

\floatstyle{ruled}
\newfloat{algorithm}{tbp}{loa}
\floatname{algorithm}{Algorithm}

\newcommand{\comment}[1]{}

\newcommand{\michal}[1]{{\color{red}{MF: #1}}}

\begin{document}

\newcommand{\be}{\begin{equation}}
\newcommand{\ee}{\end{equation}}
\newcommand{\argmin}{\mathop{\rm argmin}}
\newcommand{\argmax}{\mathop{\rm argmax}}
\newcommand{\vr}[1]{{\mathbf{#1}}}
\newcommand{\bydef}{\stackrel{\bigtriangleup}{=}}
\newcommand{\eps}{\varepsilon}
\newcommand{\mm}[1]{\mathrm{#1}}
\newcommand{\mc}[1]{\mathcal{#1}}
\newcommand{\mb}[1]{\mathbf{#1}}
\newcommand{\vect}[1]{\ensuremath{\mathbf{#1}}}
\newcommand{\R}{\mathbb{R}}

\def \EE   {{\mathbb E}}
\def \OPT {\mathcal{OPT}}
\def \vf  {\textrm{vf}}
\def \dvf {\varphi}
\def \utility {u}
\def \reals {{\mathbb R}}

\newcommand{\dist}{F}
\newcommand{\dists}{{\mathbf \dist}}
\newcommand{\distsmi}{\dists_{-i}}
\newcommand{\disti}[1][i]{{\dist_{#1}}}

\newcommand{\dens}{f}
\newcommand{\denss}{{\mathbf \dens}}
\newcommand{\densi}[1][i]{{\dens_{#1}}}

\newcommand{\agents}{N}
\newcommand{\nagent}{n}
\newcommand{\items}{M}
\newcommand{\nitem}{m}
\newcommand{\auction}{A}

\newcommand{\partition}{\Gamma}
\newcommand{\partitioni}{\Gamma_i}

\newcommand{\CWE}[0]{\textsf{CWE}}
\newcommand{\ef}[0]{envy free}
\newcommand{\EF}[0]{EF}
\newcommand{\CWEWOMC}[0]{CWE without market clearance}
\newcommand{\SW}[0]{\textsf{SW}}
\newcommand{\Revenue}[0]{\textsf{REV}}

\newcommand{\bid}{b}
\newcommand{\bids}{{\mathbf \bid}}
\newcommand{\bidsmi}{{\mathbf \bid}_{-i}}
\newcommand{\bidi}[1][i]{{\bid_{#1}}}

\newcommand{\val}{v}
\newcommand{\vals}{{\mathbf \val}}
\newcommand{\valsmi}{{\mathbf \val}_{-i}}
\newcommand{\vali}[1][i]{{\val_{#1}}}
\newcommand{\valith}[1][i]{{\val_{(#1)}}}

\newcommand{\util}{u}
\newcommand{\utils}{{\mathbf\util}}
\newcommand{\utilsmi}{\utils_{-i}}
\newcommand{\utili}[1][i]{\util_{#1}}

\newcommand{\price}{p}
\newcommand{\prices}{{\mathbf \price}}
\newcommand{\pricei}[1][i]{{\price_{#1}}}

\newcommand{\type}{t}
\newcommand{\types}{{\mathbf \type}}
\newcommand{\typesmi}{\types_{-i}}
\newcommand{\typei}[1][i]{{\type_{#1}}}

\newcommand{\alloc}{X}
\newcommand{\allocs}{{\mathbf \alloc}}
\newcommand{\allocsmi}{\allocs_{-i}}
\newcommand{\alloci}[1][i]{{\alloc_{#1}}}
\newcommand{\talloc}{Y}
\newcommand{\tallocs}{{\mathbf\talloc}}
\newcommand{\tallocsmi}{\tallocs_{-i}}
\newcommand{\talloci}[1][i]{{\talloc_{#1}}}

\newcommand{\rank}{r}
\newcommand{\ranks}{{\mathbf \rank}}
\newcommand{\ranksmi}{\ranks_{-i}}
\newcommand{\ranki}[1][i]{{\rank_{#1}}}

\newcommand{\crit}{\theta}
\newcommand{\crits}{{\mathbf \crit}}
\newcommand{\critsmi}{\crits_{-i}}
\newcommand{\criti}[1][i]{{\crit_{#1}}}

\newcommand{\decl}{d}
\newcommand{\decls}{{\mathbf \decl}}
\newcommand{\declsmi}{\decls_{-i}}
\newcommand{\decli}[1][i]{{\decl_{#1}}}

\newcommand{\demand}{D}
\newcommand{\demands}{{\mathbf \demand}}
\newcommand{\demandsmi}{\demands_{-i}}
\newcommand{\demandi}[1][i]{{\demand_{#1}}}

\newcommand{\CWEalg}{{\sc CWE} algorithm}

\newcommand{\Pool}{\text{Pool}}
\newcommand{\Pop}{\textbf{Pop}}
\newcommand{\Push}{\textbf{Push}}
\newcommand{\Bundle}{\textbf{Bundle}}
\newcommand{\ResolveConflict}{\textbf{ResolveConflict}}

\newcommand{\RaisedAgents}{N}


\newcommand{\Reject}{\text{Reject}}
\newcommand{\AllocateDemand}{\textbf{AllocateDemand}}
\newcommand{\RaisePrices}{\textbf{RaisePrices}}

\conferenceinfo{STOC'13,} {June 1--4, 2013, Palo Alto, California, USA.} 
 \CopyrightYear{2013} 
 \crdata{978-1-4503-2029-0/13/06} 
 \clubpenalty=10000 
 \widowpenalty = 10000

\title{Combinatorial Walrasian Equilibrium}

%
%
%
%
%

\numberofauthors{3} 
%
\author{
%
%
\alignauthor
Michal Feldman\\
       \affaddr{The Hebrew University of Jerusalem \& Harvard University}\\
       \email{michal.feldman@huji.ac.il}
\alignauthor
Nick Gravin\\
       \affaddr{Nanyang Technological University}\\
       \email{ngravin@pmail.ntu.edu.sg}
\alignauthor Brendan Lucier\\
       \affaddr{Microsoft Research New England}\\
       \email{brlucier@microsoft.com}
}
\date{}

\maketitle

\begin{abstract}
We study a combinatorial market design problem, where a collection of indivisible objects is to be priced and sold to potential buyers subject to equilibrium constraints.
The classic solution concept for such problems is Walrasian Equilibrium (WE), which provides a simple and transparent pricing structure that achieves optimal social welfare.  
The main weakness of the WE notion is that it exists only in very restrictive cases. To overcome this limitation, we introduce the notion of a Combinatorial Walrasian equilibium (CWE), a natural relaxation of WE. 
The difference between a CWE and a (non-combinatorial) WE is that the seller can package the items into indivisible bundles prior to sale, and the market does not necessarily clear.

We show that every valuation profile admits a CWE that obtains at least half of the optimal (unconstrained) social welfare. 
Moreover, we devise a poly-time algorithm that, given an arbitrary allocation $X$, computes a CWE that achieves at least half of the welfare of $X$.
Thus, the economic problem of finding a CWE with high social welfare reduces to the algorithmic problem 
of social-welfare approximation.
In addition, we show that every valuation profile admits a CWE that extracts a logarithmic fraction of the optimal welfare as revenue.  
Finally, these results are complemented by strong lower bounds when the seller is restricted to using item prices only, which motivates the use of bundles.
The strength of our results derives partly from their generality --- 
our results hold for arbitrary valuations
that may exhibit complex combinations of substitutes and complements.
\end{abstract}




\section{Introduction}
\label{sec:introduction}

Recent years have been marked by an explosion of interest in the role of computer science theory in market design.
Large-scale, computer-aided combinatorial markets are becoming a reality, with the FCC spectrum auctions emerging as a front-running example \cite{Ausubel2002}.
The potential outcome of this line of work is a system in which many bidders, each having complex preferences over combinations of items for sale, can express these preferences to an auction resolution algorithm that decides an appropriate outcome and payments.
Spurred forward by this vision, the computer science community has generated an entire subfield of work on developing efficient algorithms for combinatorial allocation problems \cite{Archer2001, DDDR-08, DNS-06, HKNS2004, LOS99, PSS-08}.

Much of the existing work on combinatorial auctions has focused on the desideratum of incentive compatibility, where bidders are incentivized to report their preferences truthfully to an auction resolution mechanism.
It is our view, however, that the connection between computational requirements and combinatorial market design is much broader than the design of incentive compatible mechanisms, and combinatorial extensions of complex
markets are fundamental in a wider context.
In this paper we study a classic market design problem: setting prices so that socially efficient outcomes arise when buyers select their most demanded sets. We propose a natural combinatorial extension of this problem, whereby the seller can choose to bundle objects prior to assigning prices.
We demonstrate that providing this basic operation to the seller leads to the existence of (and algorithms to find) near-optimal outcomes in settings that were previously known to suffer from severe limitations.


\vspace{-0.02in}

\paragraph{Background: Walrasian Equilibrium}
A vast literature in economic theory is dedicated to methods of assigning prices to outcomes so that a market clears in equilibrium and a socially efficient outcome arises. Suppose that we have a single seller\footnote{We note that the seller could be either a government agency wishing to maximize market efficiency, or a monopolist wishing to maximize revenue.} with a set $\items$ of $m$ items for sale, and there is a set $\agents$ of $\nagent$ buyers who have possibly complex preferences over the items, represented by a valuation function $\vali(\cdot):2^{\items} \rightarrow \R_{\geq 0}$ that maps every subset of $\items$ to a real value. A strong notion of a pricing equilibrium for such a market is as follows.  First, the seller sets prices $\{\price_j\}_{j \in \items}$ on the items for sale. Second, each buyer selects his most-desired set of items at those prices,
i.e., a set $S$ in $\argmax_S \vali(S)-\sum_{j \in S}\price_j$.
If no item is desired by more than one bidder (i.e., there is no over-demand), and all items are sold (i.e., there is no over-supply), this outcome is known as a Walrasian equilibrium (WE).


The WE solution concept is appealing:
despite competition among the agents, every buyer is maximally happy with her allocation, the market clears, and the pricing structure is natural, simple, and transparent.
Moreover, it is known that a WE, when it exists, is socially efficient; i.e., it maximizes {\em social welfare} --- the sum of buyers' valuations \cite{Bikhchandani1997}.
%
The main disadvantage of WE is that the concept is ``too good to be true" --- it is known to exist only for an extremely restrictive subset of sub-modular valuations, known as {\em gross substitutes} \cite{Gul1999}.
Since a motivating feature of combinatorial auctions is the ability to capture
complementarities in the buyers' preferences (i.e., super-additive valuation functions), this restriction limits the applicability of WE in many algorithmic mechanism design settings.

Circumventing the existence problem requires relaxing the WE notion, and several approaches can be taken with respect to this relaxation.
One approach is to allow the seller to set arbitrary bundle prices instead of item prices; i.e., set a price $p_S$ for every bundle $S$ (see, e.g., \cite{Blumrosen10,Bikhchandani2002,Parkes2000}).
This approach does lead to strong existence and efficiency results, but loses much of the simplicity and transparency that is offered by item pricing.

Another approach is to relax the requirement for market clearance (while still insisting that every buyer maximizes his utility).
This approach is natural in settings where the seller wishes to maximize some well-defined objective function, such as social welfare or revenue, and might be able to credibly leave some unsold items in the market.
This relaxation completely solves the existence problem;
indeed, an outcome in which all items are priced prohibitively high would trivially adhere to the proposed equilibrium notion.
It might, however, come at a huge social expense.
This begs the question: can a good social welfare be supported in a pricing equilibrium that relaxes the market clearance condition?
As we show in Section \ref{sec:item-pricing}, the answer is surprisingly discouraging.
In particular, even for the class of {\em fractionally subadditive} valuations \cite{Feige2006} ---
a strict subset of subadditive functions that exhibits strong substitutability among items --- the loss in social welfare can be linear in the number of items.
Relaxing market clearance, therefore, is not sufficient, and a new approach is needed.
In what follows we introduce a new equilibrium concept that captures a novel approach to the problem.

\vspace{-0.02in}

\paragraph{A New Concept: Combinatorial Walrasian Equilibrium (CWE)}
We propose to pair the notion of Walrasian equilibrium with a simple combinatorial operation, as follows.
The seller first {\em partitions} the items for sale into indivisible {\em bundles}.
This partition induces a {\em reduced market}, where individual items are no longer available, rather only the specified bundles.
This operation can be perceived as redefining the items.
Each individual bundle --- now an indivisible item --- is then assigned a price, overall pricing over the reduced market is {\em linear}, i.e., the price of a set of bundles equals the sum of the bundles' prices\footnote{This essential feature --- linearity of prices in the reduced market --- distinguishes the proposed solution concept from previous notions in the spirit of bundle pricing \cite{Bikhchandani2002}.}.
The outcome of such a process is a CWE if every bidder obtains a utility-maximizing set of bundles in the reduced market, and no conflict arises.
Thus, the essential feature of a CWE is the ability of the seller to redefine his items by pre-bundling them prior to sale.
This is an innocuous and natural power to afford the seller; after all, as the owner of the objects to be sold, it seems reasonable that he may choose to repackage them as he sees fit.

Clearly, a CWE is guaranteed to exist for every valuation profile, even without relaxing market clearance.
Indeed, the seller could simply collect all objects together into a single grand bundle, and then sell that one bundle to the bidder who values it most.
However, this may be a very inefficient outcome for the market.
The natural combinatorial question, then, is whether there exists a partition of the objects (and associated prices)
so that a CWE exists and has high social welfare.
An additional question is how best to partition the objects and set prices in order to maximize the seller's revenue.
We are interested in both the existential and computational aspects of these problems.

\vspace{-0.02in}

\paragraph{Our Results}

An observation that facilitates our analysis is a characterization of the set of CWE allocations (i.e., allocations that admit supporting CWE prices).
Specifically, an outcome can be implemented at CWE if and only if that outcome is an optimal solution to a certain linear program: the configuration LP for the assignment problem, restricted to the bundles in the outcome allocation.
This implies that every CWE generates an efficient allocation of the bundles that are sold, though we note that this necessary condition is not always sufficient. In particular, the optimal allocation cannot necessarily be implemented at CWE; in Section~\ref{subsec:gap} we exhibit an example where the welfare-optimal CWE attains only $\sfrac{2}{3}$ of the unconstrained optimal social welfare.

Characterization in hand, we study the problem of finding CWE outcomes that maximize social welfare and/or revenue in general market settings.
%
%
Our main result is the following:

\vspace{0.1in}

\noindent{\bf Result 1 (2-approximation for social welfare):}
Given an allocation $\tallocs$, we provide an algorithm that computes a CWE $\allocs$ that guarantees a social welfare of at least  $\frac{1}{2}\SW(\tallocs)$, and runs in polynomial time, given an access to each bidder's {\em demand oracle}.

\vspace{0.1in}

A direct corollary of the above result is that every instance admits a CWE that obtains at least half of the optimal (unconstrained) social welfare.
Note that our result does not restrict the preferences of the bidders; it holds for arbitrary valuation functions, including those with complements.
Moreover, since the result holds for arbitrary $\tallocs$, every social-welfare approximation can be converted into a CWE allocation that achieves the same approximation (up to factor 2).
In other words, our algorithm can be interpreted as a {\em black-box reduction} that reduces the economic problem of finding a CWE with good social welfare to the algorithmic problem of social-welfare approximation for a given class of valuation functions.
The fact that our method proceeds in a black-box manner is significant, as it allows a separation of the algorithmic and economic aspects of our pricing problem; such reductions have been developed only rarely, such as for approximately efficient Bayesian incentive compatible mechanisms \cite{HL10,HKM11}.

The presentation of our algorithm is given in two stages.
We first describe an algorithm that provides the desired approximation result, albeit might run in exponential time.
The advantage of this algorithm is its simplicity and its natural interpretation as an ascending price auction (see Section~\ref{sec:sw}).
A more challenging task is to modify the proposed algorithm into a poly-time algorithm that preserves the same approximation ratio.
This is the content of Section~\ref{sec:poly}.
The algorithm is polynomial, given an access to a {\em demand oracle} of each agent (in the reduced market) --- where agents get a set of item prices and respond with their most desired bundle given these prices.
We note that the definition of CWE implicitly involves agents answering demand queries, hence our assumption that we have access to demand oracles is effectively driven by the notion of CWE itself.
Additionally, since we consider reduced markets defined by the seller's choice of partition, we will allow demand queries over any reduced market (not just the original instance without bundles).

We also provide a negative result illustrating the need for bundles: there are instances in which no equilibrium with \emph{item} prices gives a sublinear approximation to social welfare, even if valuations are fractionally subadditive.

We next consider the problem of revenue maximization and provide the following results.

\vspace{0.1in}

\noindent{\bf Result 2 ($O(\log \nagent)$-approximation for revenue):}
Given an allocation $\tallocs$ to $\nagent$ buyers, we provide an algorithm that computes a CWE $\allocs$ that extracts revenue of $O(\log \nagent)$ fraction of $\SW(\tallocs)$, and runs in polynomial time, given an access to agents' {\em demand oracles}.

\vspace{0.1in}

Moreover, this result is tight in terms of the trade-off between social welfare and revenue objectives for the outcomes that might be supported at CWE: there are instances in which no CWE extracts more than a logarithmic fraction of the social welfare.  A corollary of our result is that, for any class of valuations functions that admits a polytime constant approximation to social welfare, one can find (in polytime) a CWE that obtains an $O(\log n)$ approximation to the revenue-optimal CWE.  Furthermore, a computational hardness result due to Briest \cite{Briest06} shows that one cannot hope for better than a polylogarithmic approximation: even in the special case of unit-demand bidders, where CWE reduces to envy-free pricing, there is a lower bound of $\Omega(\log^{\epsilon}(n))$ for the problem of approximating optimal revenue (subject to natural hardness assumptions).



\vspace{-0.02in}

\paragraph{Our techniques}


The aforementioned configuration LP provides a useful framework to study efficiency of stable pricing and existence of Walrasian equilibria. Generally speaking, the techniques of LP relaxations and rounding and especially the configuration LP are prevailing instruments in the study of combinatorial auctions.
However, there is a new distinguishing feature in the context of Combinatorial Walrasian Equilibrium, as we have to determine a proper partition of the market into bundles. One may think of every bundle in a partition as imposing a set of additional linear constraints on the configuration LP of the initial (unpartitioned) market. In order to meet the goal of achieving a nearly efficient and stable state, we need to cast a carefully chosen set of those new constraints which push the solution of a derived LP to be achieved in an integral point corresponding to a proper allocation.

The main combinatorial tool employed in our design and analysis resembles techniques taken from the theory of stable matching.
In particular, our scheme proceeds in a fashion that is similar to the Gale-Shapley algorithm~\cite{Gale62}, with bidders and items residing in the two sides of the market, and bidders ``making proposals" to the items.
During the procedure, the price of each item reflects the item's preferences over the bidders and it keeps growing monotonically.
Meanwhile, the choice of every bidder becomes scarcer and at greater expenses.
Finally, the resulting allocation of buyers to bundles may be viewed as a matching, since every allocated set of items and/or bundles may be further treated as a single big bundle.

Despite the similarities to the Gale-Shapley algorithm, there are several important aspects that distinguish our setup from the standard setting of stable matching.
Firstly, our combinatorial auction model allows for bidders to demand {\em sets} of items.
As a result, bundles demanded by the bidders may overlap in a complex way, which makes our task of resolving conflicts on the over-demanded items incomparably more difficult than for the unit demand valuations in any matching setup.
Secondly, the stable matching framework assumes no money in the market, while in our setting prices play a crucial role to guarantee stability. Finally, our routine begins with an initial allocation which is provided as part of the input and serves as a benchmark against which to compare the obtained social welfare.
The initial allocation is indeed necessary if one is looking for an efficient implementation, due to the strong NP-hardness results on social-welfare approximation in combinatorial auctions.

The ascending-price nature of our basic algorithm leads to a potentially exponential runtime, as prices may climb slowly toward a stable profile.  To address this problem, we must aggressively raise prices to ``interesting'' breakpoints.  We then analyze the structure of the agents' demands at these maximal price profiles, and find that by resolving the demands of agents in a particular order we can ensure that steady progress is made toward a final solution, leading to a polynomial runtime.

To construct a CWE with high revenue, a natural approach is to impose \emph{reserves}: lower bounds on bundle prices.  However, manipulating prices in this way can affect the structure of a final equilibrium in non-trivial ways, so that it is not clear that revenue will ultimately increase.  To circumvent this issue, we begin with a CWE with high welfare, then modify prices by adding a constant amount to the price of \emph{each} bundle.  This operation is conceptually similar to imposing a reserve, but does not fundamentally change the structure of a stable allocation.  Our approach to maximizing revenue then reduces to tuning the extent of this flat price increase.

\vspace{-0.02in}

\paragraph{Related Work}
The study of pricing equilibria in markets and related concepts of outcome fairness have a rich history in theoretical economics.  Some of the earliest work in this spirit of envy-freeness is due to Foley \cite{Foley1967} and Varian \cite{Varian1974}.
An envy-free outcome is one where no agent wishes to exchange outcomes with another.
The line of work on market-clearing prices in our market assignment problem was initiated by Shapley and Shubik \cite{Shapley1971}.  Characterizations of existence of Walrasian equilibria were studied in, for example, \cite{Aumann1975,Kelso1982,Leonard1983,Bikhchandani1997,Gul1999}.


An alternative line of work considers markets with non-linear bundle prices.
Such package auctions were formalized by Bikhchandani and Ostroy \cite{Bikhchandani2002}.  Applications to combinatorial auctions include mechanisms due to Ausubel and Milgrom \cite{Ausubel2002}, Wurman and Wellman \cite{Wurman2000}, and Parkes and Ungar \cite{Parkes2000}.
Our notion of CWE differs in that the seller commits to a partition of the objects, then sets linear prices over those bundles.




The problem of computing revenue-optimal envy-free prices has received recent attention in the computer science literature.  Guruswami et al. \cite{Guruswami2005} provide approximation algorithms for envy-free pricing in certain special cases, leading to a line of work improving on the attainable approximation factors \cite{Balcan2008, Cheung2008, Hartline2011} and a polylogarithmic lower bound \cite{Briest06}.
Mu'alem \cite{MuAlem2009} studies the revenue maximization question for agents with general types.  The notion of envy-freeness has also been applied to problems in machine scheduling \cite{Cohen2010}.


Fiat et al. \cite{Fiat2009} considered an extension of envy-freeness in which no agent envies any subset of other agents.
This concept is related to our notion of CWE.
However, crucially, they restrict their definition to agents with single-minded types, which dampens the distinction between multi-envy freeness and envy-freeness. 



A significant line of work in the algorithmic mechanism design literature is concerned with the development of truthful mechanisms for combinatorial markets.  See, for example, \cite{LOS99, Archer2001, HKNS2004, DDDR-08, DNS-06, PSS-08} and references therein.  The goal in this work is to develop algorithms that elicit truthful value revelation from the bidders.  In contrast, we assume a full-information model
and our goal is to develop an algorithmic pricing structure that satisfies certain transparency and fairness conditions.

Some of our algorithms make use of \emph{demand queries}, a manner of eliciting preference information from bidders with complex valuations.
For representative works on the power of demand queries, 
see \cite{Blumrosen2005, Mirrokni2008, Dobzinski2011, Oren2012}. 


Fu, Kleinberg and Lavi \cite{FKL-12} introduced the notion of \emph{conditional equilibrium} as a WE relaxation, where no buyer wishes to add additional items to his allocation under the given prices, but may wish to drop ones. They show that, when buyers have submodular valuations, a conditional equilibrium always exists and every conditional equilibrium achieves at least half of the optimal social welfare.  While their work is similar in spirit to the results herein, our equilibrium concept differs fundamentally in that it does not relax the requirement that every agent receives a bundle in his demand set. In particular, the conditional equilibrium notion is quite weak from the agents' happiness perspective. In particular, it violates basic envy-freeness conditions, even for sub-modular valuations.


\section{Model and preliminaries}
\label{sec:model}

%

We consider an auction framework with a set $\items$ of $\nitem$ indivisible objects and a set of $n$ agents. Each agent has a valuation function $\vali(\cdot) : 2^\items \to \reals_{\geq 0}$ that indicates his value for
every set of objects, is non-decreasing (i.e., $\vali(S) \leq \vali(T)$ for every $S \subseteq T \subseteq \items$) and is normalized so that $\vali(\emptyset) = 0$.
The profile of agent valuations is denoted by $\vals=(\val_1,\dotsc,\val_\nagent)$, and an auction setting is defined by a tuple $\auction=(\items,\vals)$.

A price vector $\prices=(\price_1, \dotsc, \price_\nitem)$ consists of a price $\price_j$ for each object $j \in M$. An {\em allocation} is a vector of sets $\allocs = (\alloc_0, \alloc_1, \dotsc, \alloc_\nagent)$, where $\alloc_i \cap \alloc_k = \emptyset$ for every $i\neq k$, and $\bigcup_{i=0}^{\nagent}\alloci = \items$. In the allocation $\allocs$, for every $i\in\agents$, $\alloci$ is the bundle assigned to agent $i$, and $\alloc_0$ is the set of unallocated objects; i.e., $\alloc_0 = \items \setminus \bigcup_{i=1}^{\nagent}\alloci$.

As standard, we assume that each agent has a quasi-linear utility function. That is, if agent $i$ is allocated bundle $\alloc_i$ under prices $\prices$, then the utility of agent $i$ is $\utili(\alloci, \prices) = \vali(\alloci) - \sum_{j \in \alloci}\price_j.$
Given prices $\prices$, the {\em demand correspondence} $\demandi(\prices)$ of agent $i$ contains the sets of objects that maximize agent $i$'s utility:
\[
\demandi(\prices) = \left\{S^*:  S^* \in\argmax_{S\subseteq M}\{\utili(S,\prices)\}\right\}.
\]
A tuple $(\allocs,\prices)$ is said to be {\em stable} for auction $\auction=(M,\vals)$
if $\alloci \in\demandi(\prices)$ for every $i\in\agents$.
A price vector $\prices$ is {\em stable} if there exists an allocation $\allocs$ such that $(\allocs,\prices)$ is stable.
An allocation $\allocs$ is {\em stable} if there exists a pricing $\prices$ such that $(\allocs,\prices)$ is stable.

For a partition $\partition=(\partition_1, \ldots, \partition_k)$ of the item set $M$ we slightly
abuse notation and denote by $\partition=\{\partition_1, \ldots, \partition_k\}$ the reduced set of items, where the
valuation of each agent $i$ of a subset $S \subseteq \partition$ is $\vali(\bigcup_{j:\partition_j \in S}\partition_j)$. We
denote by $\auction_{\partition}$ an auction over the reduced set of items $\partition$ with the induced valuation profile.

Every allocation $\allocs$ induces a partition of the objects, $\partition(\allocs) = (\alloc_0, \ldots, \alloc_\nagent)$, where $\alloc_0$
denotes the unallocated objects.
A tuple $(\allocs,\prices)$, where $\allocs=(\alloc_0, \ldots, \alloc_\nagent)$, and $\pricei$ is the price of $\alloci$ for every $\alloci\neq\emptyset$, is a {\em Combinatorial Walrasian Equilibrium} (CWE) if $(\allocs,\prices)$ is stable in the auction $\auction_{\partition(\allocs)}$.
Clearly, $\alloc_0$ may be an empty set, in which case no item remains unallocated (i.e., the market clears). Allowing for $\alloc_0$ to be non-empty is essentially the relaxation of the market clearance condition. An allocation $\allocs$ is said to be CWE if it admits a price vector $\prices \in \reals_{\ge 0}^{\nagent+1}$ such that $(\allocs,\prices)$ is CWE. 

\paragraph{Relation to WE}

A tuple $(\allocs,\prices)$, where $\allocs=(\alloc_0,\alloc_1, \ldots, \alloc_n)$ and $\prices=(\price_1, \ldots, \price_m)$, is a {\em Walrasian equilibrium} (WE) if $(\allocs,\prices)$ is stable in $\auction$ and $\price_j=0$ for every item $j \in \alloc_0$.
When the latter condition is satisfied, we also say that $(\allocs,\prices)$ {\em clears the market}.
We emphasize that a CWE is weaker than a WE in two ways:
first, it allows for market reduction by means of bundling; second, it does not require market clearance, so items with positive price can be left unsold.


\subsection{Characterization of WE}

We will make use of the following characterization of an allocation that can be supported in a WE~\cite{Bikhchandani1997}.
For a given partition $\partition$ of the objects, the allocation of $\partition$ to $\agents$ can be specified by a set of
integral variables $y_{_{i,S}}\in \{0,1\}$, where $y_{_{i,S}}=1$ if the set $S\subseteq\partition$ is allocated to agent $i\in\agents$ and
$y_{_{i,S}}=0$ otherwise.
These variables should satisfy the following conditions: $\sum_S y_{_{i,S}}\leq 1$ for every $i \in \agents$
(each agent is allocated to at most one bundle) and $\sum\limits_{i, S \supseteq \partition_j} y_{_{i,S}} \leq 1$ for every $\partition_j \in \partition$ (each element of the partition is allocated to at most one agent).
A \emph{fractional allocation} of $\partition$ is given by variables $y_{_{i,S}} \in [0,1]$ that satisfy the same conditions and
intuitively might be viewed as an allocation of divisible items.
The configuration LP for $\auction_{\partition}$ is given by the following linear program, which computes
the fractional allocation that maximizes social welfare.
\begin{eqnarray}
\max & & \sum_{i,S} \vali(S)\cdot y_{_{i,S}} \nonumber\\
\mbox{s.t.} & & \sum_S y_{_{i,S}} \leq 1 \mbox{ for every } i \in \agents \nonumber\\
& & \sum_{i, S \supseteq \partition_j} y_{_{i,S}} \leq 1 \mbox{ for every } \partition_j \in \partition \nonumber\\
& & y_{_{i,S}} \in [0,1] \mbox{ for every } i \in \agents, S \subseteq \partition \nonumber
\end{eqnarray}
The characterization given in \cite{Bikhchandani1997} states that a WE exists if and only if the optimal fractional
solution to the allocation LP occurs at an integral solution.

%

\subsection{Stable item pricing}
\label{sec:item-pricing}

The proposed concept of CWE is weaker than the concept of WE both in that it allows to restrict the item set by bundling and in that it does not require market clearance. Clearly, relaxing any one of these conditions is sufficient to guarantee existence of a stable allocation.
This begs the question whether it is possible to achieve good guarantees on the market efficiency (social welfare) by relaxing only one of these conditions.
In this section, we show that relaxing the market clearance condition alone is not sufficient, in general.
In particular, for several families of valuation functions, we establish strong lower bounds on the social-welfare approximation that can be achieved in a stable item-priced allocation.
These results reinforce the need for bundling, as captured by the CWE notion.

{\bf Unit-demand \& single-minded valuations.}
Consider the following auction: bidder 1 is a unit-demand agent, who values any non-empty subset of the items at $1+\epsilon$; bidder 2 is a single-minded agent, who desires the set of all items for a value of $m$.
In the optimal allocation all $m$ items must go to agent 2 resulting in a SW of $m$. However, in every stable pricing $\prices$ that supports this allocation, there exists an item $j \in [m]$ such that $\price_j\le 1$ (otherwise the set $[m]$ cannot be a demand set of agent 2).
But this in turn implies that $j$ is demanded by agent 1. Therefore, every stable allocation assigns a single item to agent 1 for a SW of $1+\epsilon$, compared to the optimal SW of $m$, and the linear gap follows.

It might not come at a surprise that item prices are not sufficient to obtain high welfare if valuations are super-additive.
After all, if items are complementary to each other, then bundling is an intuitive operation.
Surprisingly, the next example shows that a linear gap may exist even if all valuations are sub-additive.

{\bf XOS (fractionally-subadditive) valuations.}
Consider an auction with $m$ items and two agents with the following symmetric XOS valuations.
Agent 1 is unit-demand and values every subset at $1/2-\delta$, for a sufficiently small $\delta$ (that will be determined soon).
Agent 2 values any subset of size $k$ at $\max(1,k/2)$; it is easy to verify that this is an XOS valuation.
We claim that there is no stable pricing that sells more than a single item. For every $m \geq 2$, the optimal integral solution obtains a value of $m/2$ (by giving all the items to the XOS agent).
By the characterization given in \cite{Bikhchandani1997} (see also Section~\ref{sec:model}), this allocation admits a stable pricing if and only if $m/2$ is the optimal fractional solution of the corresponding configuration LP. We will now show a fractional solution that obtains value greater than $m/2$ for every $\delta < \frac{1}{2(m-1)}$.
Consider the fractional solution in which the allocation of the first (unit demand) agent is given by $y_{1,\{j\}}=1/m$ for every $j \in [m]$, and the allocation of the second (XOS) agent is given by $y_{2,\{j\}}=\frac{1}{m(m-1)}$ for every $j \in [m]$, and $y_{2,[m]}=\frac{m-2}{m-1}$.
One can easily verify that this is a feasible solution, and the welfare obtained by $\{y_{i,S}\}$ is given by $SW(y)=\frac{m}{2}+\frac{1}{2(m-1)}-\delta$,
which is greater than $\frac{m}{2}$ for every $\delta < \frac{1}{2(m-1)}$, as required. We conclude that a stable outcome can allocate at most one object, and thus the highest welfare that can be obtained in a stable allocation is $1$, resulting in a linear gap of $m/2$.

\subsection{Efficiency loss due to CWE: A lower bound}
\label{subsec:gap}
In this section we prove a lower bound on the efficiency of combinatorial Walrasian equilibria. In particular, we show that there are instances in which no CWE obtains more than a $(\sfrac{2}{3}+\epsilon)$ fraction of the optimal social welfare, for every $\eps > 0$. Consider an auction with $3$ items and $3$ bidders. Each agent $i$ has the valuation $\vali(\{i\})=1$ and $\vali(\{1,2,3\}\setminus\{i\})=2+\eps$. 
The optimal allocation has a social welfare of $3$ with each object $i\in\{1,2,3\}$ being allocated to agent $i$. 
Any CWE that supports this allocation must, in fact, be a WE (since no items are bundled and all items are allocated).
By the characterization given in \cite{Bikhchandani1997}, such a WE exists only if $3$ is the optimal fractional solution of the corresponding configuration LP, but this is not the case:
the fractional solution $y_{1,\{2,3\}}=y_{2,\{1,3\}}=y_{3,\{1,2\}}=\sfrac{1}{2}$ obtains a total value of $3+\sfrac{3\eps}{2} > 3$.
%
%
Therefore, in any CWE either one item is unsold or at least two items are bundled together. In either case the social welfare cannot exceed $2+\eps.$ Thus, no CWE in this market can approximate the optimal social welfare within a factor better than $\frac{2+\eps}{3}$ (for an arbitrary small $\eps>0$).

\section{Social welfare approximation}
\label{sec:sw}


In Section~\ref{subsec:gap} we established a lower bound of $1.5$ on the social welfare approximation that can be achieved at a CWE.
In this section we show that there always exists a CWE that gives a $2$-approximation to the optimal social welfare for arbitrary valuations.

Our algorithm is a natural extension of the t\^{a}tonnement process that is used to achieve a Walrasian equilibrium for (gross) substitutes valuations.
In a traditional t\^{a}tonnement process, prices are initially set to zero, agents iteratively respond to the current prices with their demand, and prices of over-demanded items increase monotonically.
In our case, in addition to increasing prices of over-demanded items, the market orchestrator will make use of an additional tool ---  merging bundles.
These operations will be {\em monotone}: prices increase monotonically, and bundles never break apart.

Our algorithm is defined formally as Algorithm \ref{cwe-alg}. Informally, the algorithm begins by bundling objects according to an initial allocation ($\talloc$ in the statement of Theorem \ref{th:halfSW}) and setting properly-designed initial prices (specifically, pricing every $\talloc_j$ at $\val_j(\talloc_j)/2$).
It maintains a pool of buyers who are not allocated a demanded set (initially all buyers).
The algorithm iteratively chooses a buyer from the pool and asks for his most-demanded set.
Whenever a buyer's demand set $S$ contains more than one item, the items in $S$ are bundled together (irrevocably), the bundle $S$ is allocated to him, and any agents who were allocated subsets of $S$ are deallocated and placed back in the pool. It should be noted that with our ``aggressive'' bundling, setting initial prices too low (as in standard t\^{a}tonnements) can lead to a big welfare loss. Therefore, the initial prices must be carefully chosen. If the demanded set is a singleton that is already allocated to another agent, the conflict is resolved (in subprocedure \textbf{ResolveConflict}) by gradually increasing the price of the item, until it is not demanded by one of these agents.  The algorithm terminates when all agents' demands are satisfied.


We now present a formal description of the process that is simple and intuitive, but may run in exponential time (in $m$ and $n$).
In Section \ref{sec:poly} we present a refined algorithm that obtains the same result with a poly-time implementation.


\begin{algorithm}
\caption{Simple CWE Algorithm} \label{cwe-alg}

\textbf{Input: } Valuations $\vals$; initial allocation $\tallocs=(\talloc_1,\dotsc,\talloc_\nagent)$ \\
\textbf{Output: } A CWE $(\allocs,\prices)$

\begin{algorithmic}[1]
\STATE Initialize: $\partition = \{\talloci\ :\ \talloci \neq \emptyset\}$; $\price(\talloci) = \frac{1}{2}\vali(\talloci)$ for all $i$; $\alloci = \emptyset$ for all $i$; $\Pool = \agents$

\WHILE{$\Pool \neq \emptyset$}
   \STATE Remove an arbitrary element $a$ from $\Pool$
   \IF{$\demand_a(\prices,\partition) \neq \emptyset$}
       \STATE Choose $S\in\demand_a(\prices,\partition)$
       \STATE $\alloc_a \leftarrow S$

       \IF{$|S|>1$}
           \STATE Set $\price(S):=\sum_{\partition_j\in S}\price(\partition_j)$
	
	\FOR{$i$ such that $\alloci \in S$}
	    \STATE $\alloci \leftarrow \emptyset$
	    \STATE $\Pool \leftarrow \Pool \cup \{i\}$
	\ENDFOR
	\STATE $\partition \leftarrow \Bundle(\partition, S)$\quad $\backslash * \partition:=\{\partition_j\!:\!\partition_j\not\in\!S\}\cup\{S\}*\backslash$
       \ELSE
	\IF{$\exists b \neq a$ such that $\alloc_b = S$}
	   \STATE $\ResolveConflict(S,a,b)$
	\ENDIF
       \ENDIF
   \ENDIF
\ENDWHILE
\STATE \textbf{Return } $(\allocs,\prices)$
%
%
\end{algorithmic}

\vspace{2mm}

%

\textbf{ResolveConflict}$(S,a,b)$:
\begin{algorithmic}[1]
\WHILE{$S \in \demand_a(\prices,\partition) \cap \demand_b(\prices,\partition)$}
	\STATE $\price(S) \leftarrow \price(S) + \epsilon$
\ENDWHILE
\STATE \textbf{if} $S \not\in \demand_a(\price,\partition)$ \textbf{then} $\alloc_a \leftarrow \emptyset; \ \Pool \leftarrow \Pool \cup \{a\}$
\STATE \textbf{else} $\alloc_b \leftarrow \emptyset; \ \Pool \leftarrow \Pool \cup \{b\}$
%
\end{algorithmic}
\end{algorithm}



%
%
%
%
%
%
%
%
%
%
%
%
%
%
%
%
%
%
%

\begin{theorem}
\label{th:halfSW} Given an initial allocation $\tallocs$, 
Algorithm \ref{cwe-alg}
computes a CWE $(\allocs,\prices)$ such that $\SW(\allocs) \geq \frac{1}{2}\SW(\tallocs)$.
\end{theorem}
\begin{proof}

First, it is easy to see that the procedure must terminate.
Indeed, bundles monotonically merge and prices monotonically increase.
Moreover, on every iteration of the procedure, either a bundle merges (lines 7-12), a price strictly increases (lines 14-15 and \ResolveConflict), or the size of the pool strictly decreases (otherwise).
Thus, assuming fixed price increments of $\eps$, the algorithm is guaranteed to terminate. 


Second, upon termination, the obtained allocation and prices is a CWE.
Indeed, every agent that is removed from the pool receives his most desired bundle given the current prices,
and agents that are deallocated due to some other agent's demand go back to the pool.
Thus, when the pool becomes empty, every agent receives his demand at the current prices, and the final allocation and pricing is CWE.



It remains to show that the returned CWE obtains at least half the social welfare of the original allocation $\tallocs$.
Let $U$ be the set of agents that get non-empty allocations in $\allocs$; i.e., $U=\{i: \alloc_i \neq \emptyset\}$.
Then, the social welfare is given by
\begin{align}
\label{eq:SW}
\sum_{i\in U}\vali(\alloci)&=\sum_{i\in U}(\vali(\alloci)-\price(\alloci))+\sum_{i\in U}\price(\alloci)\nonumber\\
&=\sum_{i\in U}\utili(\alloci)+\sum_{i\in U}\price(\alloci)\nonumber\\
&\ge\sum_{i\in U}\utili(\alloci)+\sum_{i\in U}\sum_{\talloc_j\subset\alloci}\frac{1}{2}\val_j(\talloc_j),
\end{align}
where the last inequality follows directly from the following facts:
(i) every bundle $\talloc_j$ is originally priced at $\frac{1}{2}\val_j(\talloc_j)$,
(ii) the price of a bundle that is created by a merge of two bundles equals the sum of their prices, and
(iii) a bundle's price can only increase.


The second term in the RHS of \eqref{eq:SW} captures the welfare that comes from bundles $\talloc_j$ that are being allocated in $\allocs$.
We next need to take care of the welfare that comes from $\talloc_j$'s that are not being allocated in $\allocs$.
We first observe that every bundle that is allocated in our procedure, keeps being allocated until termination.
This is because a bundle can be deallocated from one agent only if it is being allocated to another agent.
Therefore, every bundle $\talloc_j$ that is not allocated in $\allocs$ has never been allocated, and therefore still has a price of $\frac{1}{2}\val_j(\talloc_j)$ (as priced originally).

Given the last observation, we conclude that for every bundle $\talloc_j$ that is not allocated in $\allocs$, it must be that $j \in U$.
Indeed, if $j \not\in U$, then agent $j$ would gained a utility of $\val_j(\talloc_j)-\price(\talloc_j) \geq \val_j(\talloc_j)-\frac{1}{2}\val_j(\talloc_j) \geq 0$ from the bundle $\talloc_j$, contradicting the fact that $\emptyset\in\demand_j(\prices)$.

However, the agent $j \in U$ who has been allocated the bundle $\talloc_j$ in $\tallocs$ is allocated in the CWE $\allocs$ another bundle which is preferred by him.
This means that agent $j$'s utility from $\alloc_j$ is at least $\frac{1}{2}\val_j(\talloc_j)$.

Since all unallocated bundles $\talloc_j$ were allocated in $\tallocs$ to agents $i \in U$,
summing over all these bundles, we get
\begin{equation}
\label{eq:SWder2}
\sum_{i\in U}\utili(\alloci) \ge \sum_{\talloc_j\notin\bigcup_{i\in U}\alloci} \frac{1}{2}\val_j(\talloc_j).
\end{equation}
By plugging the last inequality in \eqref{eq:SW} we obtain
\begin{equation*}
\sum\limits_{i\in U}\vali(\alloci)\ge\frac{1}{2}\sum\limits_{\talloc_j}\val_j(\talloc_j)=\frac{1}{2}\SW(\talloc)
\end{equation*}
and the assertion of the theorem follows.
\end{proof}

Applying \CWEalg\ to the optimal allocation $\tallocs$ we derive the following corollary.
\begin{corollary}
\label{cor:halfOpt}
For every valuation profile $\vals$, there exists a CWE that obtains at least half of the optimal social welfare.
\end{corollary}


\section{Poly--time implementation}
\label{sec:poly}

Here we discuss how one could implement our procedure efficiently in a polynomial number of demand queries to the agents. We note that in our context demand queries are indeed unavoidable, as agents must know their demand sets in order just to verify whether a given outcome is stable.
We also emphasize that our procedure takes as input an initial target allocation Y; the problem of finding an initial allocation Y is outside the scope of our procedure.  However, we can think of Y as being generated from some approximation algorithm tailored to a particular class of valuations.

The polytime algorithm is given as Algorithm \ref{cwe-alg-poly}. Informally speaking, Algorithm \ref{cwe-alg-poly} makes two main changes to the more straightforward Algorithm \ref{cwe-alg}.  First, after each main iteration of the algorithm, we invoke procedure $\RaisePrices$ to increase the prices of the allocated bundles; prices are raised as much as possible without changing the demand set of any bidder with an allocation.  Second, whenever the demand set of a player from the pool is a singleton currently allocated to another bidder $b$, we immediately award the allocation to $a$.  We then ensure that $b$ is the next bidder to be considered, and allocate to $b$ a particular set from his demand correspondence, which is determined during the previous call to $\RaisePrices$.

\begin{algorithm}[H]
\caption{Polytime CWE Algorithm} \label{cwe-alg-poly}

\textbf{Input: } Valuations $\vals$; target allocation $\tallocs=(\talloc_1,\dotsc,\talloc_\nagent)$ \\
\textbf{Output: } A CWE $(\allocs,\prices)$

\begin{algorithmic}[1]
\STATE Initialize: $\partition = \{\talloci\ :\ \talloci \neq \emptyset\}$; $\price(\talloci) = \frac{1}{2}\vali(\talloci)$ for all $i$; $\alloci = \emptyset$ for all $i$; $\Pool = \agents$; $\Reject = \emptyset$; $T_i = \emptyset$ for all $i$

\WHILE{$\Pool \neq \emptyset$}
   \STATE Remove an arbitrary element $a$ from $\Pool$
   \IF{$\util_a(S,\prices) \leq 0$ for each $S \in \demand_a(\prices,\partition)$}
       \STATE $\Reject \leftarrow \Reject \cup \{a\}$
   \ELSE
       \STATE Choose $S\in\demand_a(\prices,\partition)$
       \STATE $\AllocateDemand(a,S)$
   \ENDIF
   \STATE $\RaisePrices()$
\ENDWHILE
\end{algorithmic}

\vspace{2mm}

$\AllocateDemand(a,S)$:
\begin{algorithmic}[1]
	\IF{$|S|>1$}
	   \FOR{$i$ such that $\alloci \in S$}
	      \STATE $\alloci \leftarrow \emptyset$
	      \STATE $\Pool \leftarrow \Pool \cup \{i\}$
	   \ENDFOR
	   \STATE $\partition \leftarrow \Bundle(\partition, S)$\quad $\backslash * \partition:=\{\partition_j\!:\!\partition_j\not\in\!S\}\cup\{S\}*\backslash$
	   \STATE $\price(S) \leftarrow \sum_{\partition_j\in S}\price(\partition_j)$
	   \STATE $\alloc_a \leftarrow S$
	\ELSE
	   \STATE $\alloc_a \leftarrow S$
	   \IF{$\exists b \neq a$ such that $\alloc_b = S$}
	      \STATE $\alloc_b \leftarrow \emptyset$
	      \STATE $\AllocateDemand(b, T_b)$ \quad $\backslash *$ Attempt to allocate $T_b$ (from $\RaisePrices$) to $b$. $*\backslash$
              \ENDIF
	\ENDIF
%
%
\end{algorithmic}

\vspace{2mm}

%

$\RaisePrices()$:
\begin{algorithmic}[1]
\STATE Initialize: $\RaisedAgents \leftarrow \{i \colon \alloci \neq \emptyset \}$
\WHILE{$\RaisedAgents \neq \emptyset$}
   \FOR{$i \in \RaisedAgents$}
      \STATE choose $S_i \in \demand_i(\prices,\partition\backslash\{ \alloc_j \colon j \in \RaisedAgents \})$ \quad $\backslash *$ Demand excluding allocations to agents in $\RaisedAgents$ $*\backslash$
      \STATE $d_i \leftarrow \utili(\alloci,\prices) - \utili(S_i,\prices)$
   \ENDFOR
   \STATE $a \leftarrow \argmin_{i\in \RaisedAgents} d_i$
   \FOR{$i \in \RaisedAgents$}
      \STATE $\price(\alloc_i) \leftarrow \price(\alloc_i) + d_a$
   \ENDFOR
   \STATE $T_a \leftarrow S_a$ \quad\quad\quad $\backslash *$ $T_a$ is the set demanded by $a$, excluding allocations to agents in $\RaisedAgents$ $*\backslash$
   \STATE $\RaisedAgents \leftarrow \RaisedAgents - \{a\}$
\ENDWHILE
%
\end{algorithmic}
\end{algorithm}

We first note that, like Algorithm \ref{cwe-alg}, Algorithm \ref{cwe-alg-poly} is monotone in the following sense:

\noindent{\scshape (Monotonicity).} {\em Over the course of Algorithm \ref{cwe-alg-poly}, prices only increase and no bundle is ever split.  Moreover, once a bundle is allocated it never become unallocated.}

Also like Algorithm \ref{cwe-alg}, after an invocation of AllocateDemand completes, each bidder that isn't in $\Pool$ is allocated a demanded set.

\begin{lemma}
\label{lem.demanded}
After a call to $\AllocateDemand$ terminates, each bidder $i \not\in \Pool$ is allocated to his most demanded set.
\end{lemma}
\begin{proof}
If $\AllocateDemand$ is called with bidder $a$, then we have two cases.  If the demanded set for $a$ is $S$ with $|S| > 1$, then $a$ is allocated $S$ and other conflicting bidders are added to Pool, so the result holds inductively.  If $|S| = 1$, then $a$ is allocated $S$, so in particular $a$ is allocated his most demanded set. The call to $\AllocateDemand$ then terminates only if there is no conflicting bidder; in this case the result holds. Note that we have not yet argued that $\AllocateDemand$ will, in fact, terminate.
\end{proof}

We next show that $\RaisePrices$ increases the price vector $\prices$ to a maximal vector such that each bidder $i$ with $\alloci \neq \emptyset$ has $\alloci \in \demand_i(\prices,\partition)$.
\begin{lemma}[Correctness of RaisePrices]
\label{lem.raiseprices}
After each call to $\RaisePrices()$, for each $i$ with $\alloci \neq \emptyset$,
$\pricei(\alloci)$ is the maximal value such that $\alloci \in \demand_i(\prices,\partition)$.
\end{lemma}
\begin{proof}
For a subset of players $\RaisedAgents$ and an allocation $\allocs$, write $\partition_{\RaisedAgents}$ for $\{ \alloc_j \colon j \in \RaisedAgents \}$ and $\partition_{\neg \RaisedAgents}$ for $\partition \backslash \partition_{\RaisedAgents}$.

We first show that $\RaisePrices$ is equivalent to a different procedure which does not run in polynomial time.  In this alternative procedure, the set $\RaisedAgents$ is defined as before, and the prices of elements of $\partition_{\RaisedAgents}$ are raised uniformly and continuously until the threshold at which the demand set of some $a \in \RaisedAgents$ changes (note that this must occur eventually; the new demanded set may be $\emptyset$).  When this occurs, $a$ is removed from $\RaisedAgents$, and the prices continue to increase for the elements remaining in $\partition_{\RaisedAgents}$.  This process continues until $\RaisedAgents$ is empty.

To see that this is equivalent to $\RaisePrices$, consider some iteration of this new process, say with initial price vector $\prices$ and set $\RaisedAgents$.  Suppose the demand set of some $a \in \RaisedAgents$ changes to $S$, and that the price vector at the point of the change is $\prices'$.  We claim that $S \subseteq \partition_{\neg \RaisedAgents}$.  The reason is that any $S$ that includes elements of $\partition_\RaisedAgents$ has its price increase by at least as much as $\alloc_a$, and hence $a$ cannot prefer it to $\alloc_a$ at prices $\prices'$.  Thus when the demand set of $a$ changes, it must be to some $S \in \demand_a(\prices', \partition_{\neg \RaisedAgents}) = \demand_a(\prices, \partition_{\neg \RaisedAgents})$.
At the point at which the demand of $a$ changes, it must be that $\util_a(\alloc_a,\prices') = \util_a(S,\prices') = \util_a(S,\prices)$.  Thus the price increase between $\prices$ and $\prices'$ is precisely $\util_a(\alloc_a,\prices) - \util_a(S,\prices)$, and moreover player $a$ is precisely the player in $\RaisedAgents$ for which this quantity is minimal (since $a$ was the first for whom the demanded set changed).  It is therefore equivalent to directly compute this quantity for each player in $\RaisedAgents$, choose the minimum, and raise the price of each object in $\partition_{\RaisedAgents}$ by this amount.  This is precisely what is done by $\RaisePrices$, and hence the procedures are equivalent as required.

The lemma now follows easily from the definition of this equivalent process.  For each $i$ with $\alloci \neq \emptyset$, we have that at the point when $i$ is removed from $\RaisedAgents$, an increase of $\price(\alloci)$ would cause $\alloci \not\in \demand_i(\prices,\partition)$.  Moreover, one element in $\demand_i(\prices,\partition)$ is contained in $\partition_{\neg \RaisedAgents}$, and the price of this element does not change between the point at which $i$ is removed from $\items$ and the conclusion of the process.  Thus, when the process concludes, it will still be that an increase of $\price(\alloci)$ would cause $\alloci \not\in \demand_i(\prices,\partition)$.  Thus $\price(\alloci)$ is maximal such that $\alloci \in \demand_i(\prices,\partition)$, as required.
\end{proof}

Note that $\RaisePrices$ defines an ordering over the players with $\alloci \neq \emptyset$: the order in which they are removed from $\RaisedAgents$.  Given an iteration of Algorithm \ref{cwe-alg-poly}, we will write $\pi$ for this permutation defined by the invocation of $\RaisePrices$ on the previous iteration.  That is, $\pi(i)$ denotes the order in which player $i$ was removed; for notational convenience we will set $\pi(i) = \infty$ for all $i$ with $\alloci = \emptyset$.  Note that on the first iteration of Algorithm \ref{cwe-alg-poly} we have $\pi(i) = \infty$ for all $i$.

We now bound the number of iterations that can occur on a single invocation of $\AllocateDemand$.

\begin{lemma}
An invocation of $\AllocateDemand$ can recurse at most $n$ times.
\end{lemma}
\begin{proof}
Note that $\AllocateDemand$ concludes with a potential tail recursion, which can be thought of as an iteration of $\AllocateDemand$ with a different agent.  We must show that this tail recursion cannot occur more than $n$ times in a single invocation of $\AllocateDemand$.  To show this, we'll show that if $\AllocateDemand$ on input $a$ results in a tail recursion with input $b$, then it must be that $\pi(b) < \pi(a)$.  In particular, this means that no agent $i$ can be passed as input to $\AllocateDemand$ more than once in a recursive chain, and hence the number of recursive calls is at most $n$.

To prove the claim, note that a recursive call occurs precisely when agent $a$ demands a single object from $\partition$, and this bundle is currently assigned to a bidder $b$.  In an initial (i.e.\ non-recursive) call to $\AllocateDemand$ we have $\pi(a) = \infty$ (since $a$ was drawn from $\Pool$) and $\pi(b) < \infty$, so the result holds trivially.  In a recursive call we have $\pi(a) < \infty$, and $\alloc_b \in \demand_a(\prices,\partition)$.  However, recalling our notation from the proof of Lemma~\ref{lem.raiseprices}, we know that the demanded set $T_a$ of $a$ is contained entirely in $\partition_{\neg \RaisedAgents}$.  Thus, since $T_a = \alloc_b$, it must be that $b \not\in \RaisedAgents$ when $a$ is removed from $\RaisedAgents$, and hence $\pi(b) < \pi(a)$.
\end{proof}

\begin{theorem}
Algorithm \ref{cwe-alg-poly} runs in polynomial time.
\end{theorem}
\begin{proof}
In each iteration of the main loop, either a set of objects is bundled or an agent is added to the rejection set $R$.  Each of these can happen at most $n$ times.  Since each invocation of $\AllocateDemand$ also runs in polynomial time by the above lemma, the result follows.
\end{proof}

\begin{theorem}
Algorithm~\ref{cwe-alg-poly} returns an CWE with social welfare at least half of the optimal allocation.
\end{theorem}
\begin{proof}
The fact that Algorithm~\ref{cwe-alg-poly} returns an CWE follows immediately from Lemma~\ref{lem.demanded}. The argument for the approximation factor guarantee is the same as for Algorithm~\ref{cwe-alg}, as this depends only on the starting condition (which is unchanged) and the fact that no object becomes unallocated after it has been allocated.
\end{proof}

\section{Revenue approximation}
\label{sec:revenue}

%
%



In this section we consider the objective of the seller's revenue.
Clearly, for any valuation profile, the seller's revenue can never exceed the social welfare of the optimal allocation.
Therefore, given an allocation, its social welfare serves as a natural benchmark for the revenue objective.
We prove that given an allocation $\tallocs$, the seller can compute in polynomial time a CWE that extract revenue of $\sfrac{1}{O(\log\nagent)}$ of the social welfare of $\tallocs$.

We first prove a lower bound: there are instances in which no CWE extracts revenue greater than $\sfrac{1}{\ln(\nagent)}$ times the optimal social welfare.

{\bf An example with a $\log n$ separation.}
Consider a market that consists of $\nagent$ items and $\nagent$ unit-demand buyers, where buyer $i$ has value $\vali(\{j\})=\sfrac{1}{i}$ for every item $j$. In any optimal allocation every agent gets exactly one item, which results in a social welfare of $\sum_{i=1}^{\nagent}\sfrac{1}{i}\approx \ln\nagent$. Any reduced set of items has the same structure of the agent's valuations as before; i.e. the reduced market contains $\nitem\le\nagent$ items with $\nagent$ unit-demand buyers, where buyer $i$ has value  $\vali(\{j\})=\sfrac{1}{i}$ for every item $j$ in the market. It is easy to verify that due to the structure of valuations, in any CWE all allocated items must have the same price.
Suppose that $k$ agents receive non-empty bundles; then one of these agents has index $i \geq k$.
For this agent, $\vali(\cdot)=\sfrac{1}{i}\le\sfrac{1}{k}$.
Therefore, the price on every sold item is at most $\sfrac{1}{k}$, which generates revenue of at most $k\cdot\sfrac{1}{k}=1$.

\medskip

We next show that given an allocation $\tallocs$, one can compute a CWE that extracts revenue within a factor $\sfrac{1}{O(\log\nagent)}$ of the social welfare of $\tallocs$. As a corollary, there always exists a CWE in which the revenue is at least a $\sfrac{1}{O(\log\nagent)}$ fraction of the optimal social welfare.

To see how to construct a CWE with high revenue, consider beginning with a CWE with high social welfare.  A natural approach to increasing revenue is to impose reserve prices: a lower bound on the price of each bundle.  However, manipulating prices in this way can affect demanded sets in non-trivial ways, and it is not clear that the final outcome will actually generate more revenue (or even be stable at all).  Instead of imposing a reserve price, we will instead consider adding a constant amount to the price of \emph{each} bundle.  This operation is conceptually similar to imposing a reserve, but does not change the structure of a stable allocation (beyond compelling some agents to leave empty-handed).  We prove that there exists at least one choice for this per-bundle price increase such that the corresponding revenue is least a logarithmic fraction of the initial social welfare.


\begin{theorem}
\label{th:revenue}
Given an arbitrary allocation $\tallocs$, one can find a CWE that
extracts revenue within factor $\sfrac{1}{O(\log\nagent)}$ of $\SW(\tallocs)$ in a polynomial number of demand queries.
\end{theorem}
\begin{proof}
Given an allocation $\tallocs$, we first run Algorithm~\ref{cwe-alg-poly} with $\tallocs$ as an input, and obtain a CWE $(\allocs,\prices)$ such that $\SW(\allocs) \geq \frac{1}{2}\SW(\tallocs)$. This step can be done in a polynomial number of demand queries, as established in Section~\ref{sec:poly}.
Let $\allocs=(\alloc_0, \alloc_1, \ldots, \alloc_k)$ and $\prices=(\price_0, \ldots, \price_k)$;
that is, in the CWE $(\allocs, \prices)$, for every $i=1, \ldots, k$, agent $i$ receives the bundle $\alloc_i$ at a price of $\price_i$.
We next make the following important observation.


\begin{claim}
\label{cl:cwe-monotonicity}
Let $(\allocs,\prices)$ be a CWE, where $\allocs=(\alloc_0, \ldots, \alloc_k)$ and $\prices=(\price_0,\dots,\price_k)$.
For any positive constant $\sigma$ let $\prices^{\sigma}$ be the price vector
$(\price_0+\sigma,\dots,\price_k+\sigma)$, and $\allocs^{\sigma}$ be the allocation 
\[
\forall i=1,\dots,k \quad
\alloci^{\sigma} = \begin{cases} \alloci &\mbox{if } \vali(\alloci) \geq \price^{\sigma} \\
\emptyset & \mbox{otherwise}. \end{cases}
\]
Then, $(\allocs^{\sigma},\prices^{\sigma})$ is a CWE.
\end{claim}
\begin{proof}
For any non-empty set $S$ it holds that
$\vali(S)-\sum_{j\in S}\price^{\sigma}_j\le\vali(S)-\sum_{j\in S}\price_j-\sigma$.
On the other hand, $\vali(\alloci)-\price_i^{\sigma}=\vali(\alloci)-\price_i - \sigma$.
Since $(\allocs, \prices)$ is a CWE, it follows that $\vali(\alloci)-\price_i \geq \vali(S)-\sum_{j\in S}\price_j$ for every $S$.
Combining the above inequalities, we get that $\utili(\alloci,\price^{\sigma}) \geq \utili(S,\price^{\sigma})$.
In addition, $\utili(\alloci,\price^{\sigma}) \geq 0$ if and only if $\vali(\alloci) \geq \price_i^{\sigma}$.
The assertion follows.
\end{proof}

Let $\SW_0=\sum_{i=1}^{k}\vali(\alloci)$ denote the social welfare of CWE $(\allocs,\prices)$.
In addition, let $\ell=\lceil\log(2k)\rceil$, and for every integer $t\in\{1,\dots,\ell+1\}$ define $\sigma^{(t)}= 2^{t-1}\frac{\SW_0}{2k}$.
Let $\prices^{(t)}$ and $\allocs^{(t)}$ be the vectors with $\alloc_i^{\sigma^{(t)}}$ defined as in Claim~\ref{cl:cwe-monotonicity}
\begin{align*}
\prices^{(t)}=(\price_0+\sigma^{(t)},\ldots,\price_k+\sigma^{(t)}),
\allocs^{(t)}=(\alloc_1^{\sigma^{(t)}},\ldots,\alloc_k^{\sigma^{(t)}}),
\end{align*}
Due to Claim~\ref{cl:cwe-monotonicity}, for every $t \in \{1, \ldots, \ell+1\}$, $(\allocs^{(t)},\prices^{(t)})$ is a CWE.
For every CWE $(\allocs^{(t)},\prices^{(t)})$, we let $\SW_t$, $\Revenue_t$ and $W_t$ denote its social welfare, revenue, and the set of indices of allocated bundles, respectively.
Note that for every $t$, $W_{t+1} \subseteq W_t$.
Finally, let $\Revenue_0$ denote the revenue of CWE $(\allocs,\prices)$.
The following is the key lemma in the proof of the theorem.
\begin{lemma}
\label{lem:cwe-revenue-gap}
$\exists~~t \in \{0, 1, \dotsc, \ell+1\}\quad\text{s.t.}~~\Revenue_t\ge\frac{\SW_0}{8\ell}.$
\end{lemma}
\begin{proof}
We first observe that
\begin{align*}
\SW_0  = \sum_{i=1}^{k}\vali(\alloci)  & =  \sum_{i\in W_1}\vali(\alloci) + \sum_{i\notin W_1}\vali(\alloci) \\
& \le  \sum_{i\in W_1}\vali(\alloci) + \sum_{i\notin W_1}(\pricei + \frac{\SW_0}{2k}) \\
&\le \SW_1 + \Revenue_0 + k\cdot\frac{\SW_0}{2k}\\
&= \SW_1 + \Revenue_0 + \frac{\SW_0}{2}.
\end{align*}
The first inequality follows from the fact that for every $i\notin W_1$, $\vali(\alloci) \leq \price_i^{(1)}$, and the second inequality follows by substituting $\SW_1=\sum_{i\in W_1}\vali(\alloci)$ and  $\sum_{i\notin W_1}\pricei \leq \Revenue_0$.

Therefore, $\SW_1\ge \frac{1}{2}\SW_0 - \Revenue_0$.
One may assume that $\Revenue_0\le\frac{1}{4}\SW_0$, since otherwise the assertion of the lemma follows directly.
Thus, $\SW_1\ge \frac{1}{4}\SW_0.$

We next show that $\SW_{\ell+1}=0$.
For every $i \in \{1, \ldots, k\}$, $\price_i^{(\ell+1)} = \pricei + 2^{\ell} \frac{\SW_0}{2k} \geq \pricei + \SW_0 \geq \pricei + \vali(\alloci) \geq \vali(\alloci)$;
thus $W_{\ell+1}=\emptyset$ and $\SW_{\ell+1}=0$.

Given that $\SW_1\ge \frac{1}{4}\SW_0$ and $\SW_{\ell+1}=0$, there must exist some $t\in\{1,\ldots,\ell\}$ such that $\SW_t-\SW_{t+1}\ge \frac{\SW_0}{4\ell}.$
We get:
\begin{align*}
\frac{\SW_0}{4 \ell}  & \leq  \SW_t - SW_{t+1}\\  
&= \sum_{i\in W_{t}\setminus W_{t+1}}\vali(\alloci)\\
&\le \sum_{i\in W_{t}\setminus W_{t+1}}\left(\pricei+2^t\frac{\SW_0}{2k}\right)\\
&\le \sum_{i\in W_{t}}\left(\pricei+2^t\frac{\SW_0}{2k}\right)\\
&\le 2 \sum_{i\in W_{t}}\left(\pricei+2^{t-1}\frac{\SW_0}{2k}\right)
= 2 \Revenue_t,
\end{align*}
where the second inequality follows from the fact that for every $i\in W_{t}\setminus W_{t+1}$, $\vali(\alloci) \leq \pricei^{(t+1)}$.
We get that $\Revenue_t\ge\frac{\SW_0}{8\ell}$, completing the proof of Lemma \ref{lem:cwe-revenue-gap}.
\end{proof}

We are now ready to complete the proof of Theorem \ref{th:revenue}.
Recall that $\SW_0$ is within factor $\frac{1}{2}$ of $\SW(\tallocs)$.
Combining this with the last lemma, and noting that $\ell = \lceil \log (2k)\rceil$, where $k \leq n$ implies that
$\Revenue_t$ is within factor $\frac{1}{O(\log\nagent)}$ of $\SW(\tallocs)$.
To conclude the proof we observe that after the execution of Algorithm~\ref{cwe-alg-poly} all the quantities $\vali(\alloci)$, $\SW_t$, $\Revenue_t$ can be easily computed in polynomial time.\qed
\end{proof}

By applying Theorem~\ref{th:revenue} with the initial allocation $\tallocs$ being the optimal allocation,
the following corollary follows.

\begin{corollary}
\label{cor:revenue}
For every valuation profile $\vals$, there exists a CWE that extracts revenue within a factor $\sfrac{1}{O(\log\nagent)}$ of the optimal social welfare.
\end{corollary}
%
%
%
%

\section{Open Problems}

Our results leave many open questions and avenues for future research.

First, there is a gap between the 2-approximation result for social welfare and the lower bound of 3/2.
We conjecture that 3/2 is the true bound, but closing this gap seems a challenging task.
The integrality gap for the configuration LP approaches 2 when the integral solution is a matching, so this technique cannot be used to improve the gap.

Second, the equilibrium notion studied in this paper does not require market clearance. An important question is whether our results extend to the stronger equilibrium notion of CWE with market clearance.

Third, it would be interesting to study how well item-pricing equilibria, without market clearance, can approximate social welfare.
\emph{A priori}, it seems intuitive that bundling should be most helpful when agent valuations exhibit complementarities.
However, our negative result 
for fractionally sub-additive valuations, from Section \ref{sec:item-pricing}, shows that item pricing can be insufficient for achieving a good welfare approximation even when valuations are complement-free.
For the more restricted class of gross-substitutes valuations, a WE always exists, and thus optimal welfare can be achieved in equilibrium.
Identifying the family of valuations for which a constant fraction of the optimal SW can be achieved via item-pricing equilibrium would be an interesting research direction.
This question seems especially interesting for the class of submodular valuations.

Fourth, our algorithm receives an initial allocation as input and returns a CWE allocation that performs well with respect to the given allocation.
This result leaves open the design of a natural process that arrives at a good approximation without receiving an initial allocation.
The severe NP-hardness results for various valuation families preclude the possibility of a poly-time process that would work for arbitrary inputs, but some families of valuations (e.g. submodular) seem particularly appealing in this context.

Fifth, our algorithms make use of demand queries, but of a special form: we allow the seller to define a partition of the objects into bundles, and can then ask demand queries with respect to the reduced market in which these bundles are the objects for sale.  This interpretation of demand queries seems natural, since the assumption that agents can determine their demands is not tied to the particular items for sale.  What is the additional computational power afforded by this (seemingly stronger) definition of demand queries?

Finally, this paper operates in the full information regime, where incentive compatibility is not a concern.
An interesting question is to what extent our results can be extended to private-information settings.
That is, what is the best social welfare that can be obtained by an incentive-compatible CWE mechanism, poly-time or not?



\bibliographystyle{abbrv}



%


\end{document}